\documentclass[10pt, final, journal, letterpaper, twocolumn]{IEEEtran}
\pdfoutput=1
\usepackage{graphicx}
\usepackage{epstopdf}
\makeatother 
\usepackage{amsmath}
\usepackage{amsthm}
\usepackage{mathrsfs}
\usepackage{color}
\usepackage{amsfonts}
\usepackage{times}
\usepackage{cite}
\usepackage{lettrine}
\allowdisplaybreaks[4]
\usepackage{array}
\usepackage{amssymb}

\usepackage{stfloats}
\usepackage{footnote}
\usepackage{booktabs}
\usepackage{array}
\usepackage{algorithm}
\usepackage{algorithmic}
\usepackage{subeqnarray}
\usepackage{cases}
\usepackage{threeparttable}
\usepackage{color}
\usepackage{url}
\usepackage{bm}
\usepackage{enumerate}

\newtheorem{theorem}{\underline{Theorem}}
\newtheorem{lemma}{\underline{Lemma}}

\begin{document}
	
\title{Data Sensing and Offloading in Edge Computing Networks: TDMA or NOMA?}
	
\author{Zezu Liang,~\IEEEmembership{Student Member,~IEEE}, Hanbiao Chen,~\IEEEmembership{Student Member,~IEEE}, \\
Yuan Liu,~\IEEEmembership{Senior Member,~IEEE}, and Fangjiong Chen,~\IEEEmembership{Member,~IEEE}
\thanks{Z. Liang is with the Department of Information Engineering, The Chinese University of Hong Kong, Hong Kong (e-mail: lz017@ie.cuhk.edu.hk).}
\thanks{H. Chen, Y. Liu, and F. Chen are with the School of Electronic and Information Engineering, South China University of Technology, Guangzhou, 510641, P. R. China (email: eechb@mail.scut.edu.cn, eefjchen@scut.edu.cn, eeyliu@scut.edu.cn). }
}
\maketitle

\begin{abstract}
With the development of Internet-of-Things (IoT), we witness the explosive growth in the number of devices with sensing, computing, and communication capabilities, along with a large amount of raw data generated at the network edge. Mobile (multi-access) edge computing (MEC), acquiring and processing data at network edge (like base station (BS)) via wireless links,  has emerged as a promising technique for real-time applications. In this paper, we consider the scenario that multiple devices sense then offload data to an edge server/BS, and the offloading throughput maximization problems are studied by joint radio-and-computation resource allocation, based on time-division multiple access (TDMA) and non-orthogonal multiple access (NOMA) multiuser computation offloading. Particularly, we take the sequence of TDMA-based multiuser transmission/offloading into account. The studied problems are NP-hard and non-convex. A set of low-complexity algorithms are designed based on decomposition approach and exploration of valuable insights of problems. They are either optimal or can achieve close-to-optimal performance as shown by simulation. The comprehensive simulation results show that the sequence-optimized TDMA scheme achieves better throughput performance than the NOMA scheme, while the NOMA scheme is better under the assumptions of time-sharing strategy and the identical sensing capability of the devices.
\end{abstract}


\section{Introduction}\label{se1}
With the unprecedented development of Internet-of-Things (IoT) technology, billions of IoT devices endowed with sensing, computing, and wireless communication capabilities have been emerging, which hastens a variety of IoT applications such as smart metering \cite{VC2011}, smart home \cite{DH2010}, and smart manufacturing \cite{SVer2017}. In addition to conventional sensor networks that require the deployment of a large number of sensor nodes \cite{IF2002}, the ubiquitous usage of smartphones, wearable devices, and in-vehicle sensing devices equipped with various sensors facilitates the development of healthcare, traffic monitoring, and smart cities as they can collect data at an adequate temporal-spatial granularity \cite{ZHOU2018}. All these, sensor nodes or mobile devices, generate a large volume of data that needs to be analyzed and used for various applications. Traditionally, these devices process the data locally or transmit them to the remote cloud center. However, the sensing data become increasingly complicated and some of them even have explicit lifetime of utility, which is challenging for devices to process due to their limited resources of energy and computation \cite{KU2013}. Besides, uploading large amount of data from edge devices to the cloud will cause a bandwidth utilization bottleneck, resulting in unpredictable latency and energy consumption as they get clogged.

Fortunately, mobile (multi-access) edge computing (MEC) bridges the gap between the devices and the cloud, as it provides computing services at the network edge and thus overcomes the long backhaul latency of cloud computing \cite{MAO2017, YU2018}. Meanwhile, the migration of data processing from resource-limited devices to edge server greatly relieves the energy limitation of the devices. However, the practical implementation of MEC requires joint optimization of finite communication-and-computational resources \cite{SAR2015, chen2016, MengyuWCL, you2017, you2018, wang2017, Ymao2017, cao2019, tran2019, liang2019}. In \cite{SAR2015}, Sardellitti $\emph{et al}$ studied radio-and-computational resource allocation in both the single-user and multiuser scenarios for overall energy consumption minimization. Game-theoretic based approaches were adopted in \cite{chen2016,MengyuWCL} for achieving offloading decisions. You $\emph{et al}$ introduced finite cloud computation capacity to upper-bound the amount of sum offloaded data but neglected the cloud computing latency and computational resource allocation \cite{you2017}. This work was extended in \cite{you2018} to the case of asynchronous offloading, where the devices have heterogeneous data-arrival time instants and deadlines. \cite{wang2017} and \cite{Ymao2017} jointly considered MEC and content caching, where the sum power consumption of the devices and the MEC server was optimized in \cite{wang2017} while \cite{Ymao2017} optimized the total revenue of the network. In \cite{cao2019}, the communication-and-computational resource allocation was optimized between the user, the helper and the server based on partial offloading and binary offloading. A multi-cell, multi-server MEC system was considered in \cite{tran2019} for sum of reductions in latency and energy consumption. Moreover, Liang $\emph{et al}$ addressed two practical issues including the I/O interference in parallel computing and the overhead of result downloading phase targeting sum throughput maximization and sum energy consumption minimization \cite{liang2019}.

On the other hand, to cope with the bandwidth utilization bottleneck results from the contradiction between limited spectrum resources and vast connectivity of devices, nonorthogonal multiple access (NOMA) is emerging as a promising technology \cite{DAI2015, DING2017}. The exploitation of power domain in NOMA allows multiple devices to communicate over the same frequency-time resource simultaneously, which however introduces inter-user interference. In order to decode the signals of all devices, the base station (BS) adopts successive interference cancellation (SIC) as multiuser detection technology. Multiple studies focused on integrating NOMA and MEC to achieve enormous potential benefits. In \cite{AK2018}, the authors first proposed an edge computing aware NOMA model in which the user clustering, transmit power, radio and computation resources were optimized to minimize the energy consumption. Energy-efficient NOMA-MEC was also considered in \cite{WANG2019} and \cite{PAN2019}, where the SIC order for partial offloading and binary offloading was optimized in \cite{WANG2019} while \cite{PAN2019} applied NOMA-based transmission in both task uploading and result downloading. The analytical results in \cite{DZ2019} demonstrated NOMA can efficiently reduce the latency and energy consumption of MEC. Accordingly, lots of works considered delay minimization in NOMA-MEC \cite{DING2018, WU2018, LP2019}. In addition, NOMA-MEC has also been studied in cooperative edge computing \cite{YuwenTCOM}.

\subsection{Contribution and Organization}
In this paper, we focus on throughput maximization in a multiuser MEC system based on TDMA and NOMA offloading. Unlike prior MEC works that only consider task offloading, we consider data sensing for task generation before task offloading, leading to the three-phase of ``sensing-offloading-computing". The data-sensing phase tightly couples the later two phases, which significantly changes the scheme designs for multiuser MEC systems. It is worth noting that in practice, especially in IoT networks, mobile devices need to sense data from the environment for task generation. We adopt the ``sense-then-offload" protocol to coordinate the data sensing and offloading process, where the devices first sense data from the environment and then offload the sensed data to the edge server for computing \cite{LI2019}. This gives rise to a tradeoff between sensing and offloading. That is, a short sensing time leads to few data to offload, while a long sensing time can generate more data but reduce the time of data offloading. Moreover, there exists another important problem of transmission sequence (who transmits first?) which has not been charted before, i.e., the devices offload later can sense more data but have less offloading time, and vice versa.
Thus it requires a new design of joint radio-and-computational resource allocation.
The main contributions of this paper are summarized as follows:
\begin{itemize}
\item In the TDMA-based offloading scheme, we consider the sequence of multiuser offloading to improve the system performance. To our best of knowledge, this is the first attempt that optimizes the sequence of TDMA-based offloading in MEC. We formulate an NP-hard problem of jointly optimizing the offloading sequence, lengths of sensing-offloading-computing time, and computation resource allocation, with the goal of offloading throughput maximization. First, we theoretically prove that under the assumption of identical sensing rates, the optimal offloading sequence follows the ascending order of the devices' weighted transmission rates. Based on this, the problem is then optimally solved with closed-form solutions. Moreover, the asynchronous computing case is studied for the TDMA-based scheme.
\item In the NOMA-based offloading scheme, the joint problem of sensing-offloading-computing time division, transmit power control, and computation resource allocation is studied. We show that the proposed method for the TDMA case can be extended to the NOMA case, and thus the optimal solutions are obtained by an efficient algorithm. In addition, we study a special case that all devices have the same sensing capability and adopt the time-sharing strategy to achieve  higher offloading throughput.
\item Our simulation results show that, the sequence-optimized TDMA scheme can achieve better throughput performance than the NOMA scheme. When the devices have identical sensing capability, the NOMA scheme with time-sharing strategy performs best.
\end{itemize}

The rest of this paper is organized as follows. Section \ref{se2} presents the system model and the throughput maximization problems based on TDMA and NOMA, respectively. Sections \ref{se3} and \ref{se4} propose algorithms for solving the corresponding problems. Section \ref{se5} provides extensive simulation results and discussions. Finally, Section \ref{se6} concludes the paper.

\section{System Model and Problem Formulation}\label{se2}
As shown in Fig. \ref{model}, we consider a multiuser MEC system consisting of $N$ devices and a BS integrated with an edge server. We assume that the wireless channels between the BS and devices are quasi-static flat-fading so that the channel gains are constant in each frame but can vary from one frame to another. It is also assumed that the BS has perfect knowledge of network information for central processing.

We consider a ``\emph{sense-then-offload}" protocol\footnote{The setting of simultaneous sensing and offloading requires a causality constraint between sensing and transmission, which is a complicated problem even for the single-user case. If in the multiuser case, there also exists the transmission collision problem among multiple users for the simultaneous sensing and transmitting. For simplifying system designs, we adopt the ``sense-then-transmit" protocol and leave the simultaneous case as our future work.}, where each device first senses data from the environment and then offloads the sensed data to the BS. All the devices have to complete their data sensing, offloading, and edge execution within a frame of duration $T$ so as to meet the real-time \mbox{requirement}. The whole process is divided into four phases: 1) data sensing by devices, 2) data offloading by devices based on TDMA or NOMA, 3) parallel computing at the BS, and 4) computational-result downloading from the BS to devices. Since the computation results are relatively in small sizes in general, the fourth phase of downloading is neglected in this paper. We next describe the first three phases in details.

\begin{figure}[t]
	\begin{centering}
		\includegraphics[scale=0.36]{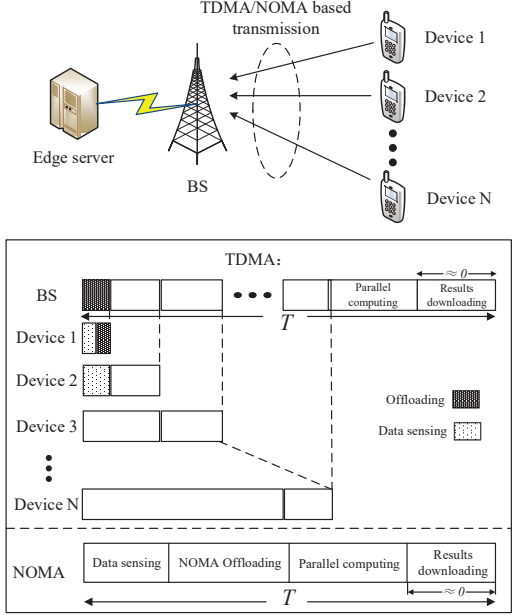}
		\caption{``Sense-then-offload" MEC systems based on TDMA and NOMA. }\label{model}
	\end{centering}
\end{figure}

\subsection{TDMA-Based Offloading}
In the TDMA case, we divide the whole frame $T$ into $N+1$ slots, each with time length denoted by $\{t_1, ..., t_N, t^c\}$ respectively. In the first $N$ slots, $N$ devices offload their sensed data to BS one by one according to a specific device-scheduling sequence, while the last slot $N+1$ is used for parallel computing at the server. For simplicity, we assume that the computing starts at the last slot, i.e., after the server receives all the sensed data from all the devices. The extension to asynchronous computing among devices will be elaborated in Section III-D. Under the above assumption, the total time constraint can be written as
\begin{align}
\text{(Time constraint in TDMA):}\quad\sum_{i=1}^{N}t_i+t^c\leq T. \label{eqn:r01}
\end{align}
As depicted in Fig. \ref{model}, each device senses data at the beginning of the frame, and continues until the slot scheduled for its data offloading arrives. The device scheduled at the first slot has to complete data sensing and offloading within the slot. After offloading is complete, devices will keep silent for saving energy. Under the sense-then-offload protocol, the scheduling sequence for devices' offloading is important for system performance, since a device scheduled at a latter slot to offload can have more time to sense data compared to those scheduled at earlier slots. To facilitate the exposition in the sequel, we define the scheduling variable $x_{n,i}$ as
\begin{equation}\label{eqn:r02}
x_{n,i}=\left\{
\begin{aligned}
& 1,\quad \text{if slot}\;i\;\text{is assigned to device}\;n\;\text{for offloading},  \\
& 0,\quad \text{otherwise}.\\
\end{aligned}
\right.
\end{equation}
The variable should satisfy the following slot assignment constraints:
\begin{align}
&\sum_{i=1}^N x_{n,i}=1,\quad \forall n, \label{eqn:r03}\\
&\sum_{n=1}^N x_{n,i}=1,\quad \forall i. \label{eqn:r04}
\end{align}

We consider a fixed sensing rate $s_n$ of device $n$. Then, given the offloading sequence $\{x_{i,n}\}$ satisfying \eqref{eqn:r03}, \eqref{eqn:r04}, the sensed data size of the device scheduled at slot $1$ is obtained as  $x_{n,1}s_n t_1^s$, where $t^s_1$ is the sensing time in the first slot, with $0\leq t_1^s\leq t_1$, and the sensed data size of the device scheduled at slot $i \leq 2$ is given by $\sum_{k=1}^{i-1}x_{n,i} s_n t_k$.

Let $P_n$, $h_n$, and $N_0$ denote the transmit power of device $n$, the channel gain between the BS and device $n$,  and the noise power, respectively. Then, the achievable uplink transmission rate of device $n$ can be expressed as
\begin{align}
r_{n}=\log_2\left(1+\frac{P_nh_n}{N_0}\right),
\end{align}
and the amount of data transmitted by device $n$ at slot $i$ is given by
\begin{align}\label{eqn:r05}
R_{n,i}=\begin{cases}
x_{n,1}(t_1-t_1^s) r_n, &\text{if~}i=1, \\
x_{n,i}t_i r_n, &\text{if~} i\geq 2.
\end{cases}
\end{align}

Based on the sense-then-offload protocol, the device can sense data before offloading, implying that the amount of transmission data is upper bounded by the amount of sensed data for each device. This introduces the following transmission data constraints:
\begin{align}
R_{n,1}&\leq x_{n,1}s_n t_1^s,\quad \forall n, \label{eqn:r06}\\
\sum_{i=2}^{N} R_{n,i}&\leq \sum_{i=2}^{N}\sum_{k=1}^{i-1}x_{n,i} s_n t_k,\quad \forall n. \label{eqn:r07}
\end{align}

After receiving the offloaded data from all devices, the BS executes them in parallel. We consider that the server has total computation resource $C$ (offloaded bits per second). Let $C_n$ denote the computation resource assigned to device $n$. Then we have the following constraints:
\begin{align}
\sum\limits_{n=1}^{N} {C_n} &\leq C, \label{eqn:r08} \\
\frac{R_{n,i}}{C_n} &\leq t^c,\quad \forall n,i. \label{eqn:r09}
\end{align}

Our objective is to maximize the sum offloading throughput of the system by joint offloading scheduling, three-phase time allocation, and the computation resource allocation. Denote $\mathbf{X}=\{x_{n,i}\}$ and $\mathbf{Y}=\{t_i, t_1^s, t^c, C_n\}$. The sum offloading throughput maximization problem can be mathematically formulated as
\begin{subequations}
	\begin{align}
	\quad{\mathbf{P1}:}~~
	\max_{\mathbf{X}, \mathbf{Y}\ge 0}\quad  & \sum_{n=1}^{N} \sum_{i=1}^{N} \omega_n R_{n,i} \\
	{\rm s.t.} \quad~\ & \sum_{i=1}^{N}t_i+t^c\leq T, \label{eqn:r10b}\\
	& \sum_{i=1}^N x_{n,i}=1,\quad \forall n, \label{eqn:r10c}\\
	& \sum_{n=1}^N x_{n,i}=1,\quad \forall i, \label{eqn:r10d}\\
	& R_{n,1}\leq x_{n,1}s_n t_1^s,\quad \forall n, \label{eqn:r10e}\\
	& \sum_{i=2}^{N} R_{n,i}\leq \sum_{i=2}^{N}\sum_{k=1}^{i-1}x_{n,i} s_n t_k,\quad \forall n, \label{eqn:r10f}\\
	&\sum\limits_{n=1}^{N} {C_n} \leq C, \label{eqn:r10g} \\
	&\frac{R_{n,i}}{C_n} \leq t^c,\quad \forall n,i, \label{eqn:r10h} \\
	& 0 \leq t_1^s \leq t_1,  \label{eqn:r10i}\\
	& x_{n,i} \in \{0,1\},\quad  \forall n,i,\label{eqn:r10j}
	\end{align}
\end{subequations}
where $\omega_n\geq 0$ denotes the weight assigned to device $n$ based on its priority and constraint \eqref{eqn:r10i} specifies the range of sensing time in the first slot.

\subsection{NOMA-Based Offloading}
In the NOMA case, the time frame $T$ is composed of three parts for all devices:
\begin{align}\label{eqn:r11}
\text{(Time constraint in NOMA):}\quad t^s+t^o+t^c\leq T,
\end{align}
where $t^s$, $t^o$, $t^c$ are the time of sensing, offloading, and computing, respectively. For simplicity, we assume the identical sensing time and offloading time for each device in the NOMA case\footnote{To accommodate the case of different sensing time and offloading time among devices \cite{NOMA_revision}, we can divide the total sensing-and-offloading time into $N+1$ time slots, where the $i$-th scheduled device senses data during the slots $1, \cdots, i$ and can offload data at slots $i+1, \cdots, N+1$. For each slot $i+1$, the scheduled devices $1, \cdots, i$ use NOMA to offload data concurrently. In this case, the device scheduling order, time-slot division, as well as the devices' power control in each slot, also need to be considered for efficient offloading designs. We leave the asynchronous case study as our future work.}. The reason is that if a device starts or finishes transmission earlier than other devices, all devices are unable to achieve a steady rate during the transmission process \cite{DING2018, LP2019}.

During the NOMA-based offloading, all the devices transmit their sensed data simultaneously and BS decodes the superimposed message based on the SIC order. Without loss of generality, we assume that $h_1\geq h_2\geq ... \geq h_N$ and the NOMA decoding follows the descending order of $h_n$'s. That is, when decoding the signal of device $n$, all other devices ($n+1, \cdots , N$) are treated as interference. Thus, the achievable uplink transmission rate of the $n$-th device is expressed as
\begin{equation}\label{eqn:r12}
r_{n}=\log_{2}\left(1+\frac{P_{n}h_{n}}{N_0+\sum_{k=n+1}^{N}P_k h_k}\right).
\end{equation}

According to the sense-then-offload protocol, the optimal amount of offloaded data of each device should be equal to the amount of data it senses. Otherwise, there exists sensed data that is not offloaded, which wastes time resource to sense and degrades the system performance. Therefore, for optimality, the amount of sensed data and offloaded data should meet
\begin{align}
s_n t^s = r_n t^o, \quad \forall n. \label{eqn:r13}
\end{align}
Since sensing rates $s_n$'s are different among devices, to satisfy conditions \eqref{eqn:r13}, the transmit power control of $P_n$ is needed. With this, the optimal transmit power of device $n$ is derived as (see Appendix A)\footnote{In the TDMA case, the optimal power control of each device for throughput maximization is to transmit at the peak power.}
\begin{align}
P_n\left(\frac{t^s}{t^o}\right)=\frac{N_0}{h_n}(2^{\frac{s_n t^s}{t^o}}-1)2^{\sum\limits_{k=n+1}^N\frac{s_k t^s}{t^o}}. \label{eqn:r14}
\end{align}

The edge computation capacity constraint in NOMA system is the same as \eqref{eqn:r08} while the  computing time $t^c$ should satisfy
\begin{align}
\quad &\frac{t^o r_{n}}{C_n} \leq t^c,\quad \forall n. \label{eqn:r15}
\end{align}

Based on the optimal power control in \eqref{eqn:r14}, the sum offloading throughput is equivalent to the sum amount of sensed data $\sum_{n=1}^N s_n t^s$. Similar to the TDMA case, the sum offloading throughput maximization problem based on NOMA transmission can be formulated as
\begin{subequations}
	\begin{align}
	{\mathbf{P2}:}~~
	\max_{t^s, t^o, t^c, C_n, \forall n} ~~  & \sum\limits_{n=1}^N s_n t^s  \\
	{\rm s.t.}\quad\quad  ~~ & t^s+t^o+t^c\leq T, \label{eqn:r16b} \\
	& 0 \leq P_n\left(\frac{t^s}{t^o}\right) \leq \overline{P_n}, \quad \forall n,\label{eqn:r16c} \\
	& \sum\limits_{n=1}^{N} {C_n} \leq C, \label{eqn:r16d} \\
	& \frac{t^o r_{n}}{C_n} \leq t^c,\quad \forall n, \label{eqn:r16e}
	\end{align}
\end{subequations}
where $\overline{P_n}$ denotes the maximum transmit power of device $n$.

\section{Offloading Throughput Maximization for TDMA Case}\label{se3}
With binary variables $\bf{X}$ and non-linear constraints \eqref{eqn:r10e}, \eqref{eqn:r10f}, \eqref{eqn:r10h}, Problem (P1) is a mixed-integer nonlinear programming problem that is difficult to solve exactly. To this end, we develop an efficient algorithm to obtain a suboptimal solution, which proceeds in two stages: First, we solve the problem of joint sensing-and-offloading resource allocation for a given computation resource allocation. In the second stage, we integrate the results derived from last stage and solve the residual computation resource allocation problem.

First, given the computation resource allocation $(t^c, \{C_{n}\})$, Problem (P1) can be written as
\begin{subequations}
\begin{align}
{\mathbf{P1.1}:}~
\max_{\substack{t_1^s, t_i \ge 0, \forall i,\\ \{x_{n,i}\}}} ~ &\sum_{n=1}^{N} \sum_{i=1}^{N} \omega_n R_{n,i} \\
{\rm s.t.} ~ \quad &\sum_{i=1}^{N}t_i+t^c\leq T, \label{eqn:r18a}\\
&\sum_{i=1}^N x_{n,i}=1,\quad \forall n, \label{eqn:r18b}\\
&\sum_{n=1}^N x_{n,i}=1,\quad \forall i, \label{eqn:r18c}\\
&R_{n,1}\leq x_{n,1}s_n t_1^s,\quad \forall n, \label{eqn:r18d}\\
&\sum_{i=2}^{N} R_{n,i}\leq \sum_{i=2}^{N}\sum_{k=1}^{i-1}x_{n,i} s_n t_k,\quad \forall n, \label{eqn:r18e}\\
&0 \leq t_1^s \leq t_1,  \label{eqn:r18f}\\
&x_{n,i} \in \{0,1\},\quad  \forall n,i. \label{eqn:r18g}
\end{align}
\end{subequations}
Clearly, Problem (P1.1) is a mixed integer linear programming (MILP) problem, which allows more tractable analysis on the sensing-and-offloading time allocation $(t_1^s, \{t_i\})$ and device offloading sequence $\{x_{n,i}\}$. In the following two subsections, we derive the closed-form optimal $(t_1^s, \{t_i\})$ and propose an effective solution of $\{x_{n,i}\}$,  respectively.

\subsection{Closed-Form Solution to Problem (P1.1) for a Given Offloading Sequence}\label{se3A}
With given $\{x_{n,i}\}$ satisfying \eqref{eqn:r18b}, \eqref{eqn:r18c}, we can find that Problem (P1.1) is a linear programming (LP) problem with a bounded feasible region. Then, there exists an optimal solution located at an extreme point \cite{book}. According to the extreme point definition, there are $N+1$ linearly independent active constraints at $\mathcal{T}$, where $\mathcal{T}\triangleq \{t_1^s, t_1, t_2, \cdots, t_N\}\in \mathbb{R}^{N+1}$. For notation simplicity, we use $n_i$ to denote the index of the device that is scheduled at slot $i$ for the given $\{x_{n,i}\}$,  i.e., $x_{n_i, i}=1$, $\forall i$. Then, the closed-form optimal $\mathcal{T}$ for a given $\{x_{n,i}\}$ can be derived by letting the $N+1$ linearly independent time constraint \eqref{eqn:r18a} and transmission data constraints \eqref{eqn:r18d} and \eqref{eqn:r18e} being active, which is given as follows:
\begin{align}
t_1^s&=\frac{(T-t^c)\prod_{j=1}^N r_{n_j}}{\prod_{j=1}^N (s_{n_j}+r_{n_j})}, \label{eqn:r18}\\
t_1&=\frac{(T-t^c)\prod_{j=2}^N r_{n_j}}{\prod_{j=2}^N (s_{n_j}+r_{n_j})}, \label{eqn:r19}\\
t_i&=\frac{(T-t^c)s_{n_i}\prod_{j=i+1}^N r_{n_j}}{\prod_{j=i}^N (s_{n_j}+r_{n_j})},\quad 2 \leq i \leq N-1, \label{eqn:r20}\\
t_N&=\frac{(T-t^c)s_{n_N}}{(s_{n_N}+r_{n_N})}, \label{eqn:r21}
\end{align}
where $s_{n_i} = \sum_{n=1}^N x_{n,i} s_n$ and $r_{n_i}=\sum_{n=1}^N x_{n,i}r_n$ denote the sensing rate and the transmission rate of the device that is scheduled at slot $i$, respectively.

\subsection{Sequence of Multiuser Offloading}
In above subsection, we derive the optimal sensing and offloading time allocation $(t_1^s, \{t_i\})$ for a given offloading sequence $\{x_{n,i}\}$. In this subsection, we propose an effective device-offloading sequence $\{x_{n,i}\}$ that is based on the following important observation of Problem (P1).
\begin{theorem} \label{th2}
Consider the case of homogeneous sensing rates $s_1=s_2=\cdots = s_N$. Then, the optimal device offloading sequence for Problem (P1) follows the ascending order of the device's weighted transmission rate $\omega_n r_n$; or equivalently, the optimal scheduling variable $x_{n,i}$ is given by
\begin{align}
x_{n,i}=\begin{cases}
1, &\text{$\omega_n r_n$ is the $i$-th smallest}   \label{eqn:r101}\\
0, &\text{otherwise}
\end{cases},  \quad \forall n, i.
\end{align}
\end{theorem}
\begin{proof}
See Appendix B.
\end{proof}

Theorem \ref{th2} reveals that under the assumption of the identical sensing rates, the optimal offloading sequence follows the ascending order of the weighted transmission rate $\omega_n r_n$, which is independent of all the other variables and thereby greatly simplifies the problem.  This is consistent with the intuition that, to offload more sensed data by device scheduling, the devices with higher transmission efficiency should be scheduled later, so that more data can be sensed while guaranteeing their successful offloading. In fact, it can be shown by simulation results that Scheme \eqref{eqn:r101} is the optimal offloading sequence, even for the general case of different sensing rates. For these reasons, we adopt Scheme \eqref{eqn:r101} as the device offloading sequence to facilitate the algorithm design.

\subsection{Computation Resource Allocation}
Based on the optimal $(t_1^s, \{t_i\})$ for given $(t^c, \{C_n\})$ and offloading sequence \eqref{eqn:r101} derived in above subsections, we develop an algorithm in this subsection to find the optimal $(t^c, \{C_n\})$ for finalizing solving Problem (P1).

Since the offloading sequence $\{x_{n,i}\}$ follows the ascending order of $\omega_n r_n$ as in \eqref{eqn:r101}, we assume that $\omega_{1}r_{1}\leq \omega_{2}r_{2}\leq \cdots \leq \omega_{N-1}r_{N-1} \leq \omega_{N}r_{N}$ without loss of generality. Then, the optimal $(t_1^s, \{t_i\})$ can be expressed in terms of $t^c$ according to \eqref{eqn:r18}-\eqref{eqn:r21}. Plugging these results, we can simplify Problem (P1) as the computation-resource allocation problem below:
\begin{subequations}
	\begin{align}
	{\mathbf{P1.2}:}~
	\max_{t^c \ge 0, C_n \ge 0, \forall n} ~&(T-t^c) \lambda  \\
	{\rm s.t.} \quad\quad  &\frac{T\mu_n}{(C_n+\mu_n)}  \leq t^c,\quad \forall n, \label{eqn:r100g} \\
	&\sum\limits_{n=1}^{N} {C_n} \leq C, \label{eqn:r100h}
	\end{align}
\end{subequations}
where $\lambda=\mu_1+\mu_2+\cdots+\mu_N$ and $\mu_n = \frac{s_{n}\prod_{j=n}^N r_j}{\prod_{j=n}^N(s_j+r_j)}$. Since $\lambda \ge 0$, it is easy to observe that $t^c = \max \limits_n \frac{T\mu_n}{(C_n+\mu_n)}$ should hold for optimality and Problem (P1.2) is equivalent to the following min max problem:
\begin{align}
	\min_{C_n \ge 0, \forall n} &\max_{n} \frac{T\mu_n}{(C_n+\mu_n)}  \nonumber \\
	{\rm s.t.} \quad &  \sum\limits_{n=1}^{N} {C_n} \leq C.
\end{align}

Due to the min-max objective function, the optimal solutions \{$C_n, t^c$\} should satisfy
\begin{align}
&\frac{T\mu_1}{(C_1+\mu_1)} =\cdots=\frac{T\mu_N}{(C_N+\mu_N)}, \label{eqn:r29a}\\
&\sum\limits_{n=1}^{N} {C_n}=C, \label{eqn:r29b}
\end{align}
By calculating \eqref{eqn:r29a} and \eqref{eqn:r29b}, we have
\begin{align}
&t^c=\frac{T \lambda}{\lambda+C} ,\label{eqn:r30}\\
&C_n=\frac{C \mu_n}{\lambda}, \quad \forall n. \label{eqn:r31}
\end{align}

By now, we have equivalently solved the original Problem (P1) in two stages. The pseudo code of this method is given in Algorithm \ref{alg:A1}.  The complexity of Algorithm \ref{alg:A1} is dominated by calculating $\{C_n\}$ and $\{t_i\}$ in Steps 2 and 3, which is of complexity $\mathcal{O}(N)$.

\begin{algorithm}[t]
	\caption{Suboptimal algorithm for Problem (P1)}\label{alg:A1}
	\label{alg:Framwork}
	\begin{algorithmic}[1]
		\STATE  Sort $\omega_n r_n, \forall n$ in an ascending order to generate the offloading order $\{x_{n,i}\}$.
		\STATE Obtain $t^c$, and $C_n$, $\forall n$, according to \eqref{eqn:r30} and \eqref{eqn:r31}.
		\STATE Obtain $t_1^s$, and $t_i$, $\forall i$, according to \eqref{eqn:r18}-\eqref{eqn:r21}.
		\ENSURE $\{x_{n,i}\}$, $t_1^s, \{t_i\}, t^c,$ and $\{C_n\}$.
	\end{algorithmic}
\end{algorithm}

\subsection{Extension to Asynchronous Computing}
In this subsection, we consider the case of asynchronous computing, i.e., for each device, the computing of sensed data starts as soon as it completes offloading. As such, for the device scheduled to offload in slot $i$, its sensed data can be computed at the BS during the sequential slots $i+1, \cdots, N+1$. Similar to the previous sections, we use index $n_i$ to denote the device scheduled at slot $i$ for the given offloading sequence $\{x_{n,i}\}$, i.e., $x_{n, i}=1$ if $n=n_i$, and $x_{n, i} =0$ otherwise. In addition, for notation briefness, we use $\{t_0, t_1, \cdots, t_N, t_{N+1}\}$ to represent the original time-slot division $\{t_1^s, t_1, \cdots, t_N, t^c\}$, where $t_0$ and $t_{N+1}$ replace $t_1^s$ and $t^c$, respectively, and $t_1$ here denotes the previous $t_1-t_1^s$, i.e., the offloading time for the first scheduled user. Then, given an offloading sequence, the sum offloading throughput maximization problem for the asynchronous-computing case can be formulated as
\begin{subequations}
\begin{align}
{\mathbf{P1.3}:} ~\max_{\substack{\{t_i\}\geq 0 \\ \{C_{n_i}\}\geq 0 }} \quad &\sum_{i=1}^N \omega_{n_i} r_{n_i} t_i  \\
\rm{s.t}~\quad~ & \sum_{i=0}^{N+1} t_i \leq T, \label{eqn:r36c1}\\
&r_{n_i} t_i \leq s_{n_i} \sum_{k=0}^{i-1} t_k, \quad \forall i, \label{eqn:r36c2}\\
&r_{n_i} t_i \leq C_{n_i} \sum_{k=i+1}^{N+1} t_k, \quad \forall i, \label{eqn:r36c3}\\
& \sum_{i=1}^N C_{n_i} \leq C. \label{eqn:r36c4}
\end{align}
\end{subequations}
We can observe that the difference of Problem (P1.3) compared to Problem (P1) is located at constraint \eqref{eqn:r36c3}, where the available computing time for the $i$-th scheduled device becomes the remaining time after offloading, rather than the original synchronous time $t^c$. Furthermore, due to the coupling between the computation resource allocation $C_{n_i}$ and time-slot division $\{t_i\}$, constraint \eqref{eqn:r36c3} is non-convex, leading to Problem (P1.3) being more challenging to solve.

In view of the fact that Problem (P1.3) with given $\{C_{n_i}\}$ is a LP problem that can be efficiently solved, we propose a method to solve Problem (P1.3) in two steps. The first step is to obtain an effective solution of $\{C_{n_i}\}$ via solving the convex relaxation problem of (P1.3), and the second step is to determine the corresponding optimal time-slot division $\{t_i\}$ by solving the residual LP problem. Since $\{t_i\}$ in the second step can be easily obtained by the interior-point method, we mainly focus on the first step of finding an effective $\{C_{n_i}\}$ in the sequel.

Specifically, we introduce a new variable $\ell_{i}$ to replace the bilinear term $C_{n_i} \sum_{k=i+1}^{N+1} t_k$ in constraint \eqref{eqn:r36c3}, and apply McCormick relaxation \cite{convex_relax} to the constraint $\ell_{i}= C_{n_i} \sum_{k=i+1}^{N+1} t_k$. Since $0\leq C_{n_i}\leq C$ and $0 \leq \sum_{k=i+1}^{N+1} t_k \leq T$ hold for all $i = 1, \cdots, N$, the McCormick relaxation on $\ell_{i}= C_{n_i} \sum_{k=i+1}^{N+1} t_k$, $\forall i$, consists of four linear constraints:
\begin{subequations}
\begin{align}
\ell_{i} &\geq 0, \label{eqn:cr1}\\
\ell_{i} &\geq C\sum_{k=i+1}^{N+1} t_k + C_{n_i} T -CT, \label{eqn:cr2}\\
\ell_{i} &\leq C\sum_{k=i+1}^{N+1} t_k,\label{eqn:cr3}\\
\ell_{i} &\leq T C_{n_i}.  \label{eqn:cr4}
\end{align}
\end{subequations}
Then, the convex-relaxed Problem (P1.3) can be written as the following form
\begin{subequations}
\begin{align}
\max_{\substack{\{t_i\}\geq 0, \{C_{n_i}\}\geq 0  \\\{\ell_i\}}} \quad &\sum_{i=1}^N \omega_{n_i} r_{n_i} t_i \\
\rm{s.t}~\quad & r_{n_i} t_i \leq \ell_i, \quad i = 1, \cdots, N,\\
& \eqref{eqn:r36c1}, \eqref{eqn:r36c2}, \eqref{eqn:r36c4},  \eqref{eqn:cr1}-\eqref{eqn:cr4}. \nonumber
\end{align}
\end{subequations}
Since Problem (38) is also LP, the optimal $\{ C_{n_i}\}$ to Problem (38) can be efficiently obtained and is also an effective solution to the original Problem (P1.3).

In terms of the device offloading sequence for Problem (P1.3), the optimal sequence can be found by exhaustive search and we use the ascending order of the weighted transmission rate as the sequence in this case. This is motivated by its simplicity for practical implementation and its effectiveness shown in the synchronous-computing case, i.e., the optimal sequence when devices' sensing rates are identical.

Based on the discussions above, we present the whole algorithm procedures of solving Problem (P1.3) in Algorithm \ref{alg:A3}. The complexity of the overall algorithm is mainly determined by solving the LP Problems (38) and (P1.3) in Steps 2 and 3, which is $\mathcal{O}(N^{3.5})$.

\begin{algorithm}[t]
	\caption{Suboptimal algorithm for Problem (P1.3)}\label{alg:A3}
	\begin{algorithmic}[1]
		\STATE \textbf{Initialization:} Set $n_i$ as the index of the device with the $i$-th smallest $\omega_n r_n$, $\forall i$.
		\STATE Solve Problem (38) and obtain its optimal solution $\{C_{n_i}^\prime\}$.
        \STATE Solve Problem (P1.3) with the given $\{C_{n_i}^\prime\}$ and obtain its optimal solution $\{t_i^\prime\}$.
        \ENSURE $\{C_{n_i}^\prime\}$ and $\{t_i^\prime\}$.
	\end{algorithmic}
\end{algorithm}

\section{Offloading Throughput Maximization for NOMA Case}  \label{se4}
In this section, we address Problem (P2) of throughput maximization for NOMA case. First, by observing the problem structure, we have the following theorem that specifies the optimality conditions for Problem (P2).
\begin{theorem}\label{th3}
The optimal $t^c, t^o, t^s$ and $C_n, \forall n$ should satisfy
\begin{align}
&t^s+t^o+t^c=T, \label{eqn:r34a} \\
&\frac{t^s}{t^o}=\beta^* \triangleq \min_{n} \beta_n^*, \label{eqn:r34b}\\
&\frac{t^o r_n}{C_n}=t^c, \quad \forall n, \label{eqn:r34c}\\
&\sum\limits_{n=1}^{N} {C_n}=C, \label{eqn:r34d}
\end{align}
\end{theorem}
where $\beta_n^*$ is the root of the following equation that can be obtained by the bisection search method:
\begin{align}
H_n(\beta)=\overline{P_{n}}-\frac{N_0}{h_n}(2^{s_n \beta}-1)2^{\sum\limits_{m=n+1}^N s_m \beta} = 0. \label{eqn:r100}
\end{align}
\begin{proof}
First, $t^s+t^o+t^c = T$ should hold in \eqref{eqn:r16b} for throughput maximization, otherwise throughput can further increase by simultaneously scaling $(t^s, t^o, t^c)$ without violating the constraints, which leads to the first condition \eqref{eqn:r34a}. Constraint \eqref{eqn:r16d} for all $n$, and constraint \eqref{eqn:r16e} also should be active, such that the computation resource is fully utilized to achieve the maximum throughput, which also can be proved by contradiction similarly. Therefore, we have conditions \eqref{eqn:r34c} and \eqref{eqn:r34d} for the optimal computation resource allocation.

On the other hand, we can observe that the objective function of Problem (P2) is monotonically increasing with $t^s$, while the device transmit power $P_n$ in \eqref{eqn:r16c} is monotonically increasing with $\frac{t^s}{t^o}$. Thus, for optimality, one of the constraints \eqref{eqn:r16c} should be active; otherwise we can   further decrease $t^o$ and at the same time increase $t^s$ as a way to improve the throughout. Let $\beta_n^*$ denote the root of $P_n(\frac{t^s}{t^o}) = \overline{P_n}$, $\forall n$. Due to the increasing monotonicity, one active constraint among \eqref{eqn:r16c} implies that the optimal $\frac{t^s}{t^o}$ should be the minimum $\beta_n^*$, which completes the proof.
\end{proof}
Theorem \ref{th3} reveals several resource allocation policies for throughput maximization in the NOMA case. First, the whole radio-and-computation resources should be fully used, specified by \eqref{eqn:r34a} and \eqref{eqn:r34d} respectively. Second, the optimal computation resource allocated to device $n$ is the minimum amount required for computing the devices' offloaded data within $t^c$. Last, under the power control strategy \eqref{eqn:r14}, the optimal tradeoff between the sensing time and offloading time is determined by the user with the minimum sensed-then-offload efficiency $\beta_n^*$, where $\beta_n^*$ depends on the user parameters quantifying sensing rate, channel, transmit power budget, and SIC order, respectively, as specified by \eqref{eqn:r100}.

Based on the results from Theorem \ref{th3}, we can derive the optimal solution to Problem (P2) by jointly solving \eqref{eqn:r34a}-\eqref{eqn:r34d}, which is given by
\begin{align}
t^s &=\left. T \middle/\left[1+\frac{1}{\beta^*} +\frac{\sum_{n=1}^{N} r_n^*}{C\beta^*}\right] \right.,  \label{eqn:r35a}\\
t^o &=t^s \cdot\frac{1}{\beta^*}, \label{eqn:r35b} \\
t^c &=t^s \cdot\frac{\sum_{n=1}^{N} r_n^*}{C\beta^*}, \label{eqn:r35c} \\
C_n &=C  \cdot\frac{r_n^*}{\sum_{n=1}^{N} r_n^*}, \quad \forall n, \label{eqn:r35d}
\end{align}
where $r_n^*$ denotes the device transmission rate in \eqref{eqn:r12} with given $P_n = P_n (\beta^*)$, $\forall n$.

We can observe that in NOMA case,  the optimal time proportion of three phases follows $t^s : t^o :t^c = 1 : \frac{1}{\beta^*} : \frac{\sum_{n=1}^{N} r_n^*}{C\beta^*}$, and the optimal computation resource allocated to device $n$ is proportional to its transmission rate.

\begin{algorithm}[t]
	\caption{Optimal algorithm for Problem (P2)}\label{alg:A2}
	\begin{algorithmic}[1]
		\FOR {each device $n$}
        \STATE \textbf{initialize} $\beta^{LB}=0$ and $\beta^{UB}>0$ with $H_n(\beta^{UB})<0$.
        \REPEAT
        \STATE Set $\beta=\frac{1}{2}\left(\beta^{LB}+\beta^{UB}\right)$.
		\IF {$H_n(\beta)>0$}
		\STATE Set $\beta^{LB}=\beta$.
		\ELSE
		\STATE Set $\beta^{UB}=\beta$.
		\ENDIF
		\UNTIL $\beta^{UB}-\beta^{LB}<\epsilon$, where $\epsilon>0$ is a very small constant for controlling accuracy.
        \STATE Set $\beta_n^* = \beta$.
		\ENDFOR
		\STATE Set $\beta^*=\min\limits_n \beta_n^*$, and obtain $t^s, t^o, t^c$, and $C_n, \forall n$, according to \eqref{eqn:r35a}-\eqref{eqn:r35b}, respectively.
        \ENSURE $t^s, t^o, t^c$, and $\{C_{n}\}$.
	\end{algorithmic}
\end{algorithm}

Finally, we summarize the whole algorithm for solving Problem (P2) in Algorithm \ref{alg:A2}. The algorithm complexity is mainly determined by the bisection search of $\{\beta_n^*\}$ in Steps 2)-12), which is $\mathcal{O} (N \log_2 \frac{1}{\epsilon})$. Here $\epsilon$ is a very small constant for controlling accuracy.

\subsection{Special Case}
In this subsection, we assume that all devices have the same sensing rates, i.e., $s_n = s_0$, $\forall n$, and consider the case of flexible SIC decoding order, known as the time-sharing strategy \cite{PD2016}. We will show that compared to the fixed SIC order case, adopting the time-sharing strategy in our NOMA scheme can achieve a higher sum offloading throughput since it can increase the minimum transmission rate among devices, i.e., improve user fairness.

The basic idea behind time sharing is to further divide the offloading time into multiple time fractions, each assigned with distinct SIC order for the BS to decode received signals. By changing the time portions among fractions, the time-sharing strategy can adjust the individual transmission rate of each user without affecting the sum rate. Let $\mathcal{M} = \{1, \cdots, m, \cdots, M\}$ denote the set of SIC orders used in the fractions, and $\pi_m (n)$ denote the order of device $n$ to be decoded in fraction $m$. For example, when $N=3$, $\mathcal{M}$ is the subset of $3!$ possible SIC orders; and if fraction $1$ follows $3 \rightarrow 1 \rightarrow 2$ to decode the devices' signals, then device $1$'s signal is at the $\pi_1(1) = 2$ order to decode. Based on these, we can express the achievable transmission rate of device $n$ in the time fraction $m$ as
\begin{align}\label{eqn:Rnm}
r_{n,m} = \log_2\left(1+ \frac{P_n h_n}{I_{n,m} + N_0}\right),
\end{align}
where $I_{n,m}\triangleq \sum_{\forall k |\pi_m(n)<\pi_m(k)\leq N }P_k h_k$ denotes the sum interference from the un-decoded devices. Similar to \cite{PD2016}, we consider that the configurations in $\mathcal{M}$ are pre-determined, and $P_n = \overline{P_n}$, $\forall n$, for sum transmission rate maximization. Hence, each $r_{n,m}$ is deterministic through \eqref{eqn:Rnm}.

Let $\tau_{m}$, with $\sum_m\tau_{m} = 1$, denote the time portion of fraction $m$. Under the homogeneous assumption of device sensing rates, the throughput maximization Problem (P2) with time-sharing can be formulated as
\begin{subequations}
\begin{align}
{\mathbf{P2.1}:}~
\max_{\substack{t^s, t^o, t^c, \\ \{\tau_m\} \{C_n\}}}\quad  & Ns_0 t^s \\
{\rm s.t.} \quad~ ~& t^s+t^o+t^c\leq T, \label{eqn:r43a} \\
& s_0 t^s \leq \sum_{m=1}^{M} \tau_m t^o r_{n,m}, \quad \forall n,\label{eqn:r43b} \\
&\!\sum_{m=1}^{M} \tau_m t^o r_{n,m} \leq C_nt^c,  \quad \forall n,\label{eqn:r43c} \\
&\!\sum_{m=1}^{M}\tau_m =1, \quad \sum_{n=1}^{N}C_n\leq C.\label{eqn:r43d}
\end{align}
\end{subequations}
Note that Problem (P2.1) can be transformed into a standard LP problem, via variable substitution on bilinear terms $\tau_m t^0$ and $C_nt^c$ and re-expression on the constraints \eqref{eqn:r43b}-\eqref{eqn:r43d} accordingly. Thus, the optimal solution to Problem (P2.1) can be efficiently obtained by the LP solver.

In the sequel, we focus on theoretical analysis of the optimal system throughput of Problem (P2.1) compared to that of Problem (P2), which is specified by the following theorem.

\begin{theorem}\label{th4}
Let $R_{\text{sharing}}$ and $R_{\text{fixed}}$ denote the optimal system throughput of Problem (P2.1) and Problem (P2), respectively. Assume that $s_n = s_0$, $\forall n$, and that $\mathcal{M}$ includes the SIC order used in Problem (P2). Then, $R_{\text{sharing}}\geq R_{\text{fixed}}$ holds for any given $\mathcal{M}$.
\end{theorem}

\begin{proof}
Let $R_{\text{sharing}}^\prime$ denote the optimal system throughput of Problem (P2.1) conditioned on all the constraint \eqref{eqn:r43b} being active. Clearly, by definition, $R_{\text{sharing}}\geq R_{\text{sharing}}^\prime$ for any given $\mathcal{M}$.
In the following, we will first express $R_{\text{fixed}}$ and $R_{\text{sharing}}^\prime$ and then show $R_{\text{sharing}}^\prime \geq R_{\text{fixed}}$ for any given $\mathcal{M}$. Based on this, $R_{\text{sharing}}\geq R_{\text{fixed}}$ naturally holds for any given $\mathcal{M}$.

To express $R_{\text{fixed}}$ and $R_{\text{sharing}}^\prime$, we need the following preliminary lemma \cite{sumrate_NOMA}:
\begin{lemma}\label{lem1}
Given the device's transmit power $P_n$, $\forall n$, the sum of the devices' transmission rates can be expressed as
\begin{align}
\sum_{n=1}^{N} r_n = \log_2 \left(1+ \frac{\sum_{n=1}^N P_n h_n}{N_0}\right),
\end{align}
which is a function of $\sum_{n=1}^N P_n h_n$ and is independent of the SIC decoding order.
\end{lemma}

First, we express $R_{\text{fixed}}$ in the case of $s_n = s_0$, $\forall n$. Combining the result of Lemma \ref{lem1} and  $s_0 t^s = r_n t^o$, $\forall n$, from \eqref{eqn:r13}, we have
\begin{align}\label{eqn:45}
\beta^* = \frac{t^s}{t^o} =  \frac{\sum_{n=1}^N r_n^*}{Ns_0} = \frac{\log_2 \left(1+ \frac{\sum_{n=1}^N P_n(\beta^*) h_n}{N_0}\right)}{Ns_0}.
\end{align}
By substituting \eqref{eqn:45} into \eqref{eqn:r35a}, $R_{\text{fixed}}$ is  derived as
\begin{align}\label{eqn:46}
&R_{\text{fixed}}=  Ns_0 t^s \nonumber \\
&=\left. T\!\middle/\!\bigg[1+\frac{Ns_0}{\log_2\big(1+ \frac{\sum_{n=1}^N P_n(\beta^*) h_n}{N_0}\big)} +\frac{Ns_0}{C}\bigg]\right..
\end{align}

Next we turn to express $R_{\text{sharing}}^\prime$. Since all constraints \eqref{eqn:r43b} are active for this case and $\sum_{m=1}^M \tau_m =1$, by Lemma \ref{lem1}, the optimal sensing-and-offloading time radio satisfies
\begin{align}\label{eqn:47}
\frac{t^s}{t^o} = \frac{\sum_{n=1}^N\!\sum_{m=1}^M\!\tau_m r_{n,m}}{N s_0} = \frac{\log_2 \left(1+ \frac{\sum_{n=1}^N \overline{P_n} h_n}{N_0}\right)}{Ns_0}.
\end{align}

Similar to Theorem \ref{th3}, it can be checked that constraints \eqref{eqn:r43a}, \eqref{eqn:r43c}, and \eqref{eqn:r43d} of Problem (P2.1) should be active for optimality. Combining these active constraints and condition \eqref{eqn:47}, we also can obtain the corresponding solution of $t^s, t^o, t^c, \{C_n\}$, and then derive $R_{\text{sharing}}^\prime$ as
\begin{align}\label{eqn:48}
R_{\text{sharing}}^\prime =  \left. T \middle/\bigg[1+\frac{Ns_0}{\log_2 \big(1+ \frac{\sum_{n=1}^N \overline{P_n} h_n}{N_0}\big)} +\frac{Ns_0}{C}\bigg]\right..
\end{align}
Since $\overline{P_n}\geq P_n(\beta^*)$, $\forall n$, we can easily observe from \eqref{eqn:46} and \eqref{eqn:47} that $R_{\text{sharing}}^\prime \geq R_{\text{fixed}}$ holds for any given $\mathcal{M}$, which completes the proof.
\end{proof}

The intuition behind Theorem \ref{th4} is that in the NOMA case with $s_n = s_0$, $\forall n$,  the system throughout maximization is equivalent to maximizing the sensing time $t^s$, and the maximum $t^s$ is further restricted by the minimum transmission rate among devices [see \eqref{eqn:r13} and \eqref{eqn:r43b}]. Compared to the fixed SIC order scheme, the time-sharing strategy can flexibly adjust the time portions among multiple SIC orders to achieve a larger minimum device transmission rate. Hence, adopting time-sharing strategy can achieve higher system throughput.

It is worth noting that the time-sharing strategy performance also depends on the set of SIC orders $\mathcal{M}$ used in the fractions. The optimal $\mathcal{M}$ is with $M = N!$, i.e., considering all the possible SIC orders among devices. However, it incurs prohibitively high complexity when the number of devices is large. For the practical implementation of the time-sharing, we can adopt the greedy algorithm proposed in \cite[Section IV]{PD2016} to obtain an efficient set $\mathcal{M}$. Its main idea is to repeatly insert the candidate orders into set $\mathcal{M}$ until the minimum transmission rate does not increase. Due to the lack of space here we omit to explain its detailed implementation.

\section{Simulation Results}\label{se5}
We present the simulation results in this section to evaluate the performance of the proposed schemes. The simulation parameters are set as follows unless specified otherwise. We set the time frame $T=1$ s, the device weights $\omega_n=1, \forall n$, the system bandwidth $B=1$ MHz, and the noise power $N_0=10^{-9}$ W. The channel power gains between the BS and all devices are modeled as $h_{n}=cd_{n}^{-\gamma}\rho_{n}^{2}, \forall n$, where $d_{n}$ denotes the distance between the BS and the $n$-th device which is distributed uniformly in the range $(0, 50]$ m. $c=10^{-3}$, $\gamma=3.5$ and $\rho_n$ represent the average signal power attenuation at a reference distance, the pathloss exponent, and a Gaussian random variable satisfying independent and identical distribution (i.i.d), respectively. For each device, the maximum transmit power is set as $P=1$ W while the sensing rate follows a uniform distribution with $s_n\in[1,10]\times 10^5$ bits/s. Last, the computational capacity of MEC units $C$ is set to be $1\times 10^7$ offloaded bits per second. The simulation results are obtained by averaging over 500 realizations.

\subsection{Throughput Maximization Based on TDMA}
For performance comparison, we also evaluate the performance of the following benchmark schemes. Note that all these benchmark schemes are used to find the offloading sequence $\{x_{n,i}\}$, in view of the fact that Problem (P1) with given $\{x_{n,i}\}$ is convex and can be efficiently solved.

\begin{enumerate}[\quad A)]
	\item \textbf{Optimal via exhaustive search:} Enumerate all the $N!$ possible offloading sequences to find the optimal one;
	\item \textbf{Random offloading order:} All the devices offload in an order that is randomly chosen from the $N!$ possible orders;
	\item \textbf{Ascending order of sensing rates:} All the devices offload in the ascending order of sensing rates (i.e., $s_n$);
	\item \textbf{Descending order of transmission rates:} All the devices offload in the descending order of transmission rates (i.e., $r_n$);
\end{enumerate}

In Fig. \ref{fig2}, we compare the sum throughput performance of different schemes when the number of devices $N$ varies from 4 to 24. First, we can observe that our proposed Algorithm \ref{alg:A1} that uses the ascending order of transmission rate (i.e., $r_n$) as the offloading sequence, can achieve the optimal performance same as the exhaustive search, indicating the effectiveness of our proposed sequence scheme. It can also be seen that the proposed Algorithm \ref{alg:A1} outperforms the B, C and D benchmark schemes as $N$ increases and obtains about $14\%$ throughput improvements over the B, C schemes and $30\%$ over the D scheme, which implies that the offloading order would have a significant impact on the system throughput. D scheme achieves the worst performance due to its entirely opposite order compared to our proposed sequence scheme. $B$ and $C$ schemes achieve similar performance because from the perspective of transmission rates, their offloading orders are both random. This also validates our conclusion that the optimal order is independent of sensing rates (i.e., $s_n$). On the other hand, we can see that compared to the synchronous case, the proposed Algorithm \ref{alg:A3} in the asynchronous computing case can further improve the system throughput especially when $N$ is large. This demonstrates the superiority of parallelism between computing and TDMA-based offloading in the large-scale system.

\begin{figure}[t]
	\begin{centering}
		\includegraphics[scale=0.5]{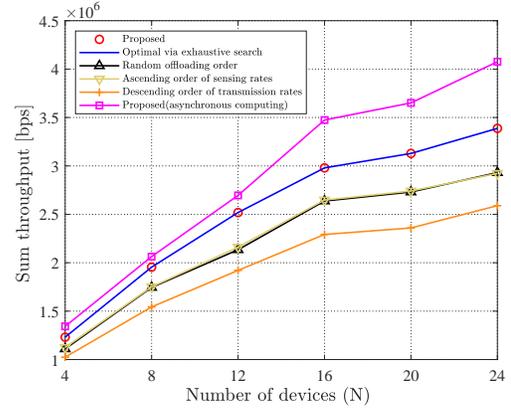}
		\vspace{-0.05cm}
		\caption{Sum system throughput  v.s. the number of devices $N$. }\label{fig2}
	\end{centering}
	\vspace{-0.2cm}
\end{figure}
\begin{figure}[t]
	\begin{centering}
		\includegraphics[scale=0.5]{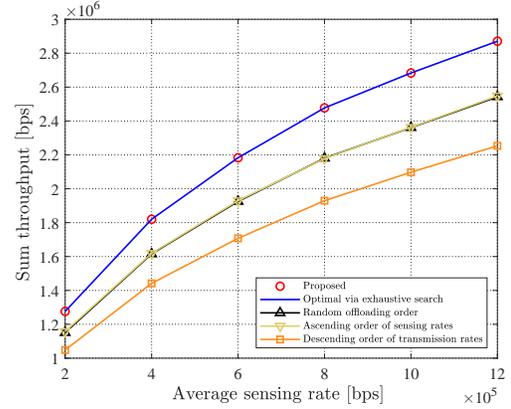}
		\vspace{-0.05cm}
		\caption{Sum system throughput  v.s. the average sensing rate }\label{sensingrate}
	\end{centering}
	\vspace{-0.2cm}
\end{figure}
In Fig. \ref{sensingrate}, we study the impact of devices' sensing rates  on the sum throughput performance, where $N = 8$, $s_n \in [1, s_\text{max}]\times 10^5$ bits/s, $\forall n$,  and the different average sensing rates are obtained by varying $s_\text{max}$ from $3$ to $23$. We can observe that no matter what the range and variance of devices' sensing rates are, the proposed Algorithm \ref{alg:A1} that uses the ascending order of $r_n$ as the device-scheduling sequence, always achieves the optimal performance same as the exhaustive search. This demonstrates the ascending order of $\omega_n r_n$ is the optimal offloading-sequence scheme, even for the general case of different sensing rates.

\subsection{Performance Comparison of Multiple Access Techniques}
A comparison of sum throughput obtained by TDMA, NOMA and FDMA will be presented in the following. Specifically, the sum throughput of TDMA and TDMA with asynchronous computing are obtained by Algorithm \ref{alg:A1} and Algorithm \ref{alg:A3}, respectively; and the sum throughput of NOMA is obtained by Algorithm \ref{alg:A2}. While for FDMA, we denote $\alpha_n$ as the proportion of bandwidth allocated to device $n$ and the sum throughput of FDMA is obtained by solving the problem:
\begin{subequations}
\begin{align}
\max_{\substack{\{\ell_n\}, t^s, t^o, t^c\\ \{B_n\}, \{C_n\}}} ~~ &  \sum_{n=1}^{N} \ell_n\\
\rm{s.t}~~\quad & t^s + t^o +t^c \leq  T, \\
&\ell_n\leq \min\bigg\{t^o \alpha_n B \log\big(1\!+\!\tfrac{P_n h_n}{\alpha_n B N_0}\big),\bigg.\nonumber \\
&\qquad \qquad \qquad\qquad \quad  \bigg.t^s s_n, ~t^c C_n \bigg\}, ~\forall n, \label{eqn:46b} \\
&\sum_{n=1}^N \alpha_n \leq 1, \quad \sum_{n=1}^N C_n \leq C.
\end{align}
\end{subequations}
where constraint \eqref{eqn:46b} means that for each device, its offloading throughput is restricted by the minimum amount of data among the phases of sensing, offloading, and computing. It is well-known that $t^o \alpha_n B \log\left(1+\tfrac{P_n h_n}{\alpha_n B N_0}\right)$ in \eqref{eqn:46b} can be converted into the concave form by variable substitution. Thus, the above problem can be transformed to a convex problem and efficiently solved by existing solvers, e.g., CVX \cite{cvx}.

\begin{figure}[t]
	\begin{centering}
		\includegraphics[scale=0.5]{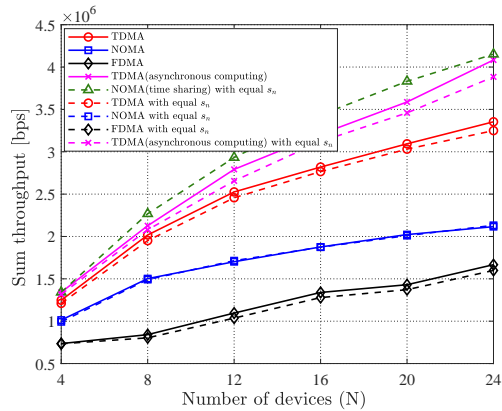}
		\vspace{-0.05cm}
		\caption{Sum system throughput v.s. the number of devices $N$. }\label{fig3}
	\end{centering}
	\vspace{-0.2cm}
\end{figure}

Fig. \ref{fig3} shows the sum throughput obtained by TDMA, TDMA with asynchronous computing, NOMA with fixed decoding order, NOMA with time sharing, and FDMA, versus different number of devices. Since NOMA with time sharing works for the situation that all devices have the same sensing rate, we also consider all the schemes in the case with equal $s_n=5\times10^5$ bit/s for comparison.
First, we can see that the sum throughput of different multiple access techniques basically follow TDMA $>$ NOMA $>$ FDMA. This is because NOMA and FDMA perform the sequential operations in sensing and offloading, while TDMA can take advantage of sensing-and-offloading parallelism, i.e., the time allocated for sensed-data offloading of one device is also utilized for sensing data at the un-scheduled devices, to make better use of time. NOMA outperforms FDMA because the multiplexing in power domain can achieve a higher sum transmission rate in the offloading phase compared to FDMA. Second, it can be observed that NOMA with time sharing can achieve a significant throughput improvement compared to NOMA with fixed SIC order, which demonstrates the benefit of the flexible SIC order scheme used in time sharing.

\begin{figure}[t]
	\begin{centering}
		\includegraphics[scale=0.5]{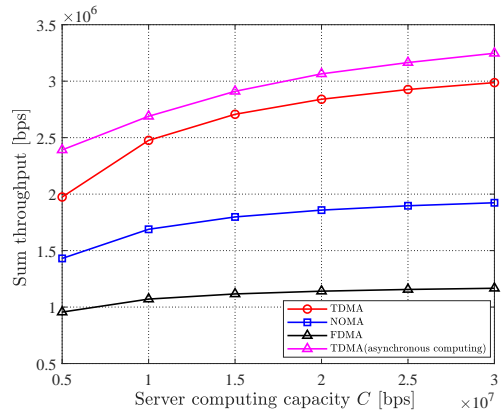}
		\vspace{-0.05cm}
		\caption{Sum system throughput  v.s. the server computing capacity $C$. }\label{fig4}
	\end{centering}
	\vspace{-0.2cm}
\end{figure}

Fig. \ref{fig4} illustrates the relationship between the sum throughput and the server computing capacity $C$, where the number of devices $N$ is set to be $12$. As expected, the sum throughput increases with computing capacity in all considered schemes. This is because that increasing computing capacity can reduce the computing time $t^c$ and prolong the time for data sensing and offloading, thus contributing to higher sum throughput.

\begin{figure}[t]
	\begin{centering}
		\includegraphics[scale=0.5]{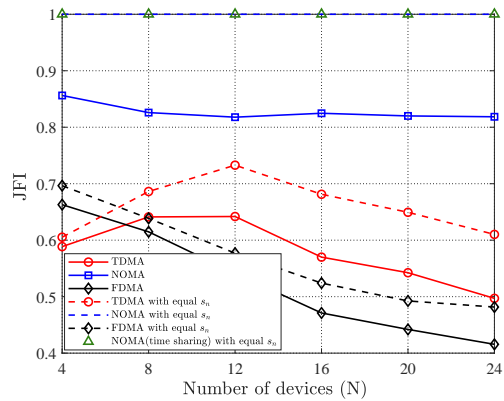}
		\vspace{-0.1cm}
		\caption{ Jain's fairness index comparison. }\label{fig5}
	\end{centering}
	\vspace{-0.1cm}
\end{figure}

Fig. \ref{fig5} illustrates the user fairness of different multiple access schemes versus the number of devices $N$. Here, we adopt Jain's fairness index (JFI) as a metric to quantify user fairness, which is defined as
\begin{align}
J=\frac{(\sum_{n=1}^N R_n)^2}{N\sum_{n=1}^N R_n^2}. \label{eqn:r45}
\end{align}
One can observe that when devices have varied sensing rates, the JFIs of TDMA,  NOMA, and FDMA, basically follow NOMA $>$ TDMA $>$ FDMA. This is because NOMA scheme follows the idea of increasing the minimum device offloading throughput to improve the sum throughput, while TDMA and FDMA prefer to increase the offloading throughput of the device with high sense-and-offload efficiency (via allocating more resources) as a way to maximum the sum throughput.
Another observation is that when $N$ increases, the JFI of TDMA scheme decreases; while for TDMA with equal $s_n$, the JFI first increases when $N\leq 12$ and begins to decrease afterwards. This could be explained by the fact that when $N\leq 12$, the devices achieve similar throughput due to the same sensing rate. However, when $N$ increases, there is greater diversity of devices' channel conditions which make it difficult to maintain the fairness. The observations from Fig. \ref{fig3} and Fig. \ref{fig5} reveal that when the devices have varied sensing rates, TDMA scheme can achieve the maximum sum throughput while NOMA scheme can promote the fairness among devices but has relative small sum throughput. When all devices have the same sensing rate, NOMA with time sharing can guarantee both the maximum sum throughput and the maximum fairness.

\section{Conclusions}\label{se6}
In this paper, we have investigated throughput maximization in a multiuser MEC system based on TDMA and NOMA offloading with a sense-then-offload protocol. The formulated problems require a joint control of data sensing, offloading and edge-server computing. A set of low-complexity algorithms are proposed for both TDMA and NOMA cases. They are either optimal or close-to-optimal by leveraging the insights of problem structure. Moreover, asynchronous computing for TDMA-based offloading and time-sharing strategy for NOMA-based offloading are also studied. The simulation results show that when the devices have different sensing rates, TDMA scheme is better the NOMA one for sum throughput maximization since the NOMA scheme sacrifices the throughput for fairness promotion. When all devices have the same sensing rate, NOMA with time-sharing strategy will serve as the optimal option since it can guarantee both the maximum sum throughput and the maximum fairness.

\begin{appendices}
\section{Proof of \eqref{eqn:r14}}\label{AP1}
According to \eqref{eqn:r12}, we can express $P_n$ as
\begin{align}
P_n=\frac{(2^{r_{n}}-1)(\sum_{k=n+1}^N P_k h_k+N_0)}{h_n}.\label{eqn:r46}
\end{align}
Rewriting $r_n$ in \eqref{eqn:r12} as $r_{n}=\log_{2}\left(\frac{\sum_{k=n}^{N}P_{k}h_{k}+N_0}{\sum_{k=n+1}^{N}P_k h_k+N_0}\right)$ yields
\begin{align}
\sum\limits_{k=n+1}^{N}r_{k}=\log_{2}\bigg(\frac{\sum_{k=n+1}^{N}P_{k}h_{k}+N_0}{N_0}\bigg).\label{eqn:r47}
\end{align}
It can be derived from \eqref{eqn:r47} that
\begin{align}
\sum\limits_{k=n+1}^N P_k h_k+N_0=N_0 2^{\sum\limits_{k=n+1}^{N}r_{k}} \label{eqn:r48}
\end{align}
Substituting \eqref{eqn:r48} and \eqref{eqn:r13} into \eqref{eqn:r46} leads to the desired result.

\section{Proof of Theorem \ref{th2}} \label{AP2}
Let $s_0$ denote the sensing rate of each device. Based on the active constraints \eqref{eqn:r18d} and \eqref{eqn:r18e},  we can rewrite the objective function of Problem (P1.2) as
\begin{align}
\sum_{n=1}^{N} \sum_{i=2}^{N}\sum_{k=1}^{i-1}x_{n,i} s_0 \omega_n t_k+\sum_{n=1}^{N} x_{n,1} s_0 \omega_n t_1^s.\label{eqn:r25}
\end{align}
For notational simplicity, we use $n_i$ to denote the device index with $x_{n_i, i}=1$ inside an offloading sequence $\{x_{n,i}\}$. Then, \eqref{eqn:r25} can be equivalently re-expressed as
\begin{align}
s_0&\left(\omega_{n_1}t_1^s+ \sum_{i=2}^N\omega_{n_i}\sum_{k=1}^{i-1}t_k\right),\label{eqn:r26}
\end{align}
where $\omega_{n_i}=\sum_{n=1}^N x_{n,i}\omega_n$ denotes the weight of the device scheduled at slot $i$. Let us start with the scenario of two devices (i.e., $N = 2$). Using the closed-form solutions derived in the last subsection, we have the sum offloading throughput as:
\begin{align}
S(2)=s_0(T-t^c)\frac{(\omega_{n_1}\!+\!\omega_{n_2})r_{n_1}r_{n_2}+s_0\omega_{n_2}r_{n_2}}{(s_0+r_{n_1})(s_0+r_{n_2})}. \label{eqn:r27}
\end{align}

We can observe that to maximize the sum throughput, the optimal scheduling order should satisfy $\omega_{n_1}r_{n_1}\leq \omega_{n_2,2}r_{n_2,2}$. Hence, Theorem \ref{th2} holds for the two-devices case, which is a starting point from which to begin the induction.

Suppose that scheduling in an ascending order of the devices' weighted transmission rates (i.e., $\omega_{n_1}r_{n_1}\leq \omega_{n_2}r_{n_2} \leq \cdots \leq \omega_{n_{N-1}}r_{n_{N-1}}$) achieves the maximum sum offloading throughput $S(N-1)$ for the $(N-1)$ devices. Based on the assumption, we intend to deduce the same conclusion also holds for the $N$-devices case. Using the induction hypothesis $S(N-1)$, we have
\begin{align}
S(N)=S(N-1)\!+s_0(T-t^c)\frac{I_1\! + \!I_2 + \!\cdots \!+\! I_{N-1}}{\prod_{i=1}^{N}(s_0+r_{n_i})}, \label{eqn:r28}
\end{align}
where
\begin{align}
&I_1\!=\!(\omega_{n_N}r_{n_N}\!-\!\omega_{n_1}r_{n_1})s_0 \!\prod_{j=2}^{N-1}\!r_{n_j}, \label{eqn:r28a} \\
&I_2\!=\!(\omega_{n_N}r_{n_N}\!-\!\omega_{n_2}r_{n_2})(s_0 r_{n_1}\!\prod_{j=3}^{N-1}\!r_{n_j}\!+\! s_0^2\! \prod_{j=3}^{N-1}\!r_{n_j} ),\label{eqn:r28b}\\
&\quad \vdots \nonumber\\
&I_{N-1}\!=\!(\omega_{n_N}r_{n_N}\!-\!\omega_{n_{N-1}}r_{n_{N-1}})\nonumber\\
&(s_0 r_{n_1}\!\prod_{j=1}^{N-2}\!r_{n_j}\!+\!\cdots\ \!+\! s_0^{N-2}r_{n_1}\!+\! s_0^{N-2}r_{n_2} \!+\! \cdots+ \nonumber\\
&s_0^{N-2}r_{n_{N-2}} \!+\! s_0^{N-1} ). \label{eqn:r28N}
\end{align}

Obviously, maximizing $S(N)$ implies that $\omega_{n_N}r_{n_N}\ge\omega_{n_i}r_{n_i}$ should be met for all $i = 1, \cdots, N-1$. Since $\omega_{n_1}r_{n_1}\leq \omega_{n_2}r_{n_2} \leq \cdots \leq \omega_{n_{N-1}}r_{n_{N-1}}$, we finally have $\omega_{n_1}r_{n_1}\leq \omega_{n_2}r_{n_2} \leq \cdots \leq \omega_{n_{N-1}}r_{n_{N-1}} \leq \omega_{n_N}r_{n_N}$. Consequently, the general result follows by the Principle of Mathematical Induction thus completes the proof.

\end{appendices}

\bibliographystyle{IEEEtran}
\bibliography{IEEEabrv,bill}

\end{document}